\documentclass[journal]{IEEEtran}

\usepackage{algorithm}
\usepackage{algorithmic}
\usepackage{amsmath}
\usepackage{amssymb}
\usepackage{latexsym}
\usepackage{multirow}
\usepackage{epsfig}
\usepackage{graphics}
\usepackage{graphicx}
\usepackage{mathrsfs}
\usepackage{subfigure}
\usepackage{slashbox}
\usepackage{mathrsfs}
\usepackage{pifont}
\usepackage{bbding}
\usepackage{stmaryrd}
\usepackage{amssymb}
\usepackage{amsfonts}
\usepackage{epic}
\usepackage{curves}
\usepackage{stfloats}
\usepackage{epstopdf}
\usepackage{cite}
\usepackage{stfloats}

\newcommand{\beq}{\begin{equation}}
\newcommand{\enq}{\end{equation}}
\newcommand{\ben}{\begin{eqnarray}}
\newcommand{\enn}{\end{eqnarray}}
\newcommand{\bei}{\begin{itemize}}
\newcommand{\eni}{\end{itemize}}

\newcommand{\bm}[1]{\mbox{\boldmath{$#1$}}}

\newtheorem{theorem}{Theorem}

\newtheorem{remark}{Remark}

\makeatletter
\newcommand{\figcaption}{\def\@captype{figure}\caption}
\newcommand{\tabcaption}{\def\@captype{table}\caption}
\makeatother
\date{}

\begin{titlepage}

\title{
\vspace{-1.0cm}
   \hfill{\em\small{}}\\
   \vspace{0.6cm} \LARGE
\begin{center}
{Achieving Optimality in Robust Joint Optimization of Linear Transceiver Design}\end{center}
}

\author{Hongying~Tang, Wen~Chen,~\IEEEmembership{Senior Member,~IEEE},\\ Jun Li, ~\IEEEmembership{Member,~IEEE} 
\thanks{Copyright (c) 2015 IEEE. Personal use of this material is permitted. However, permission to use this material for any other purposes must be obtained from the IEEE by sending a request to pubs-permissions@ieee.org.}
\thanks{Hongying~Tang and Wen~Chen are with the Department of Electronic Engineering,
Shanghai Jiaotong University, Shanghai, and School of Electronic Engineering and Automation, Guilin University of Electronic Technology, Guilin, China. (e-mail: \{lojordan, wenchen\}@sjtu.edu.cn). Jun Li is with school of Eletrical and Information Engineering, University
of Sydney, Australia. Email: jun.li1@sydney.edu.au.}
\thanks{This work is supported by the National 973 Project \#2012CB316106, NSF China \#61328101, by STCSM Science and
Technology Innovation Program \#13510711200, and by SEU National Key
Lab on Mobile Communications \#2013D11.}}

\end{titlepage}

\begin{document}

\maketitle

\begin{abstract}
This paper presents new results on linear transceiver designs in a multiple-input-multiple-output (MIMO) link. By considering the minimal total mean-square error (MSE) criterion, we  prove that the robust optimal linear  transceiver design has  a channel-diagonalizing structure, which verifies the conjecture in the previous work \cite{JW_2011}. Based on this property,  the original design problem can be transformed into a scalar problem, whose global optimal solution is first obtained in this work. Simulation results show the performance advantages of our solution over the existing schemes.

\end{abstract}

\IEEEpeerreviewmaketitle


\section{Introduction}\label{sec:1}
MIMO technique has attracted a considerable interest from both academic and industrial fields in recent years. By exploiting the multiplexing and diversity property, it can significantly improve
the spectral efficiency and link reliability of the system \cite{wang}.
In the literatures, transceiver designs in MIMO systems have been extensively studied in \cite{EB_color, EB_poisson, wang, perfect_mmse, Joint_Tx_Rx, XZ_Pre-fixed,JW_Pre-fixed, JW_2011, MD_MSE, NK_MMSE_Selection, HS_Worst_MR,HS_MMSE,thywork}. One approach of the designs is to allow nonlinear process at the transmitter or the receiver, such as the
successive
interference cancelation receiver design discussed by \cite{wang}, or the Maximum Likelihood detector investigated in \cite{EB_poisson, EB_color}.

As an alternative approach, the linear transceiver design, which only allows linear matrix multiplication of the signal,  is more preferable in a practical system due to low implementation complexity, and is the focus of this paper.
In \cite{perfect_mmse}, the joint optimal linear transceiver design problem was addressed,  and a closed-form solution was derived. Their result was generalized  into the multicarrier MIMO system in \cite{Joint_Tx_Rx}, by developing a unified optimization framework.
The aforementioned works \cite{perfect_mmse, Joint_Tx_Rx} enjoy a common favorable feature that the transceiver processing matrix parallelized the original channel and allocated power to each data stream. In light of the optimality of this \emph{channel-diagonalizing} structure in the perfect channel state information (CSI) case, one may wonder whether the same property holds for the robust design in the imperfect CSI case.

Robust design, which aims to reduce the sensitivity of the imperfect CSI  to the system performance, has attracted much attention \cite{XZ_Pre-fixed,JW_Pre-fixed, JW_2011, MD_MSE}. Generally, there are two widely used CSI uncertainty  models in the literature: the stochastic model and the deterministic model.
For the statistical CSI uncertainty model, where the distribution of CSI uncertainties is assumed to be known,  this channel-diagonalizing structure has been well established in MIMO channels \cite{XZ_Pre-fixed, MD_MSE}. However, for the deterministic CSI uncertainty model, which assumes that the instantaneous value of CSI error is norm-bounded, this problem remains unsolved, and only some restricted results were obtained in \cite{JW_Pre-fixed, JW_2011}.
The authors in \cite{JW_Pre-fixed} proposed a semi-robust scheme, by optimizing only the transmit processing matrix with some fixed equalizer. Obviously, this scheme cannot fully exploit the performance gain by the equalizer, since the fixed equalizer may not be optimal. Later in \cite{JW_2011},  the authors considered joint linear transceiver design, and showed a superior performance over \cite{JW_Pre-fixed}.  By imposing certain structural constraints on the processing matrix at the transmitter or receiver side,
they  observed the favorable channel-diagonalizing structure. Then they transformed the original problem into the issues of power loading among each data stream, which were further solved by the alternation optimization method.  However, two problems in \cite{JW_2011} were left unsolved:
{\begin{itemize}
         \item [$\mathcal Q1$)] \emph{Joint Optimal structure}:  Without any additional structural restriction,  is this channel-diagonalizing structure joint optimal?
         \item [$\mathcal Q2$)] \emph{Global Optimal solution}:  If it is, does the alternating-optimization based method converge to the global optimal solution?
\end{itemize}

In this paper, we will answer the above two questions raised by \cite{JW_2011}. Without assuming any specific structure for the linear transmitter-equalizer matrix, we show that the optimal design actually admits a channel-diagonalizing structure. Based on this property, the original problem reduces to a scalar convex problem, whose optimal solution can thus be efficiently obtained. Simulation results in section~\ref{sec:simu} show the superior performance of our solution over that in \cite{JW_2011}.

\emph{Notations}: $[\cdot]^H$ denotes conjugate transpose of a matrix or a vector. $\mathbf I$ and $\mathbf 0$ denote the identity and zero matrix, respectively. $\mathbb R^N$ and $\mathbb C^N$ respectively denote the $N$ dimensional real field and complex field. $||\cdot||_2$ and $\|\cdot\|_F$  denote the the Frobenius norm of a vector and a matrix, respectively. We will use boldface lowercase letters to denote column vectors and boldface uppercase letters to denote matrices. The positive semidefinite matrix $\mathbf X$ is denoted by $\mathbf X\succeq 0$.  diag$\{\mathbf x_1, \cdots, \mathbf x_R\}$ denotes diagonal concatenation of block matrices $\mathbf x_1, \cdots, \mathbf x_R$. The tr$(\cdot$) is the trace of a matrix. vec$(\mathbf X)$ stacks the columns of matrix $\mathbf X$ into a vector. $\otimes$ denotes the Kronecker product. $\mathfrak R\{\cdot\}$ denotes the real part of a complex number. $\lambda_{\max}(\cdot)$ is the maximum eigenvalue of a matrix.


\section{Problem Statement}\label{sec:2}
We consider  a MIMO communication system equipped with $N$ transmit antennas at the source and $M$ receive antennas at the destination. The symbol vector $\mathbf s\in \mathbb C^L$ is linearly precoded by a source precoding matrix $\mathbf F\in \mathbb C^{N\times L}$, through the MIMO channel $\mathbf H\in \mathbb C^{M\times N}$, and then received by the destination. We assume that $E\{\mathbf s\mathbf s^H\}=\mathbf I$ without loss of generality. Generally, the transmitter imposes a power constraint on the precoding matrix $\mathbf F$ as $\text{tr}(\mathbf F\mathbf F^H)\leq P$. A linear equalizer $\mathbf G\in \mathbb C^{L\times M}$ is usually applied on the received signal to obtain the estimated symbol vector $\hat{\mathbf s}$ as
\ben
\hat{\mathbf s}=\mathbf G\mathbf H\mathbf F\mathbf s+\mathbf G\mathbf n,\nonumber
\enn
where $\mathbf n\in \mathbb C^M$ is the additive white Gaussian noise (AWGN) observed at the destination with variance $\sigma_n^2\mathbf I$. Then the MSE between $\hat{\mathbf s}$ and $\mathbf s$ is given by
\ben
\text{MSE}\triangleq \mathcal E\{\|\hat{\mathbf s}-\mathbf s\|^2\}=\|\mathbf G\mathbf H\mathbf F-\mathbf I\|_F^2+\sigma_n^2\|\mathbf G\|_F^2.\nonumber
\enn
We assume that $L\leq \text{rank}(\mathbf H)$, since the number of degrees of freedom is upper bounded by $L\leq \text{rank}(\mathbf H)=\min\{M, N\}$.

In a practical wireless communication scenario, perfect CSI is usually difficult to obtain. With only imperfect CSI, the system performance will be deteriorated. This motivates us to investigate the robust design taking the CSI errors into account. To characterize the mismatched CSI, we adopt a common deterministic imperfect CSI model \cite{JW_2011, JW_Pre-fixed}, and write the channel matrix as
\ben\label{equ:roba1}
\mathbf H=\tilde {\mathbf H}+\mathbf E,
\enn
where $\tilde {\mathbf H}$ is the estimated channel matrix and $\mathbf E$ is the corresponding CSI error matrix satisfying $\|\mathbf E\|_F\leq \varepsilon$ for some $\varepsilon\geq 0$. As in \cite{JW_2011, JW_Pre-fixed}, we assume that only $\tilde {\mathbf H}$ and $\varepsilon$ are available at both ends.

By taking into imperfect CSI model \eqref{equ:roba1} into account, the robust transmitter-equalizer design is given by the solution of the following min-max problem:
\ben\label{equ:robmse}
\min_{\mathbf G, \mathbf F}\max_{\|\mathbf E\|_F\leq \varepsilon }&&\|\mathbf G(\tilde {\mathbf H}+\mathbf E)\mathbf F-\mathbf I\|_F^2+\sigma_n^2\|\mathbf G\|_F^2,\nonumber\\
\text{s.t.}&&\text{tr}(\mathbf F\mathbf F^H)\leq P.
\enn

\section{Robust Joint Optimal Structure of $\mathbf F$ and $\mathbf G$}\label{sec:opr}
In this section, we will determine the joint optimal structure of $\mathbf F$ and $\mathbf G$ in problem \eqref{equ:robmse}, showing that they diagonalize the MIMO channel into eigen subchannels.  We also figure out that the worst-case CSI uncertainty $\mathbf E$ has the similar singular value decomposition (SVD) structure as the nominal channel $\tilde {\mathbf H}$, which simplifies problem \eqref{equ:robmse} into a scalar problem as we will shown in section~\ref{sec:scalar}.

Denote the SVD  structure of $\mathbf F$ and $\mathbf G$ by $\mathbf F=\mathbf U_f\mathbf \Sigma_f\mathbf V_f^H$ and $\mathbf G=\mathbf U_g\mathbf \Sigma_g\mathbf V_g^H$, respectively, where $\mathbf U_f, \mathbf V_f, \mathbf U_g$ and $\mathbf V_g$ are unitary matrices. The matrices $\mathbf \Sigma_f$ and $\mathbf \Sigma_g$ can be written as
\ben
\mathbf \Sigma_f=[\hat{\mathbf \Sigma}_f,\mathbf 0]^T, \mathbf \Sigma_g=[\hat{\mathbf \Sigma}_g, \mathbf 0],\nonumber
\enn
where $\hat{\mathbf \Sigma}_f\triangleq \text{diag}\{f_1, \cdots, f_L\}$ and $\hat{\mathbf \Sigma}_g\triangleq \text{diag}\{g_1, \cdots, g_L\}$ are real diagonal matrices. Denote the nominal channel $\tilde{\mathbf H}$ by $\tilde{\mathbf H}=\mathbf U_h\mathbf \Sigma_h\mathbf V_h^H$ and let $\hat{\mathbf \Sigma}_h$ be the $L\times L$ diagonal matrix containing the largest $L$ singular values $\gamma_1\geq \cdots \geq \gamma_L$. Then the following theorem determines the optimal structure of $\mathbf F$ and $\mathbf G$.
\begin{theorem}\label{theorem:robmse}
The robust optimal F and G in problem \eqref{equ:robmse} can be expressed in the following structure:
\ben
\mathbf F&=&\mathbf V_h\mathbf {\Sigma}_f,\label{equ:f1}\\
\mathbf G&=&\mathbf {\Sigma}_g\mathbf U_h^H\label{equ:f2}.
\enn
Meanwhile the corresponding worst case channel uncertainty is given by  $\mathbf E=\mathbf U_h\mathbf \Delta_D\mathbf V_h^H$, with $
\mathbf \Delta_D=\text{diag}\{\hat{\mathbf \Delta}_D, \mathbf 0\}$,
and $\hat{\mathbf \Delta}_D\in \mathbb R^{L\times L}$  being diagonal.
\end{theorem}

\begin{proof}
We write the first additive term of the objective function of \eqref{equ:robmse} as
\ben\label{equ:kun}
\|\mathbf G(\hat{\mathbf H}+\mathbf E)\mathbf F-\mathbf I\|_F^2
&\overset{(a)}=&\|\mathbf G\mathbf U_h(\mathbf \Sigma_h+\mathbf \Delta)\mathbf V_h^H\mathbf F-\mathbf I\|_F^2\nonumber\\
&\overset{(b)}=&\|\mathbf G'(\mathbf \Sigma_h+\mathbf \Delta)\mathbf F'-\mathbf I\|_F^2\nonumber,
\enn
where in $(a)$ we have defined $\mathbf \Delta\triangleq \mathbf U_h^H\mathbf E\mathbf V_h$.  By the unitary-invariant property of $\|\cdot\|_F$, we can see that $\mathbf \Delta$ still satisfies $\|\mathbf \Delta\|_F\leq \varepsilon$. In $(b)$ we have defined $\mathbf G'\triangleq\mathbf G\mathbf U_h$ and $\mathbf F'\triangleq\mathbf V_h^H\mathbf F$.
Now problem \eqref{equ:robmse} can be rewritten as
\ben\label{equ:robmse2}
\min_{\mathbf G', \mathbf F'}\max_{\|\mathbf \Delta\|_F\leq \varepsilon }&&\text{MSE}(\mathbf F', \mathbf G', \mathbf \Delta)\nonumber\\&&\triangleq\|\mathbf G'(\mathbf \Sigma_{h}+\mathbf \Delta)\mathbf F'
-\mathbf I\|_F^2+\sigma_{d}^2\|\mathbf G'\|_F^2,\nonumber\\
\text{s.t.}&&\text{tr}(\mathbf F'\mathbf F'^H)\leq P,
\enn
which is an optimization problem with respect to $\mathbf F'$ and $\mathbf G'$.
To proceed, we first discuss a particular case when $(\mathbf F', \mathbf G')=([\hat{\mathbf \Sigma}_{f},\mathbf 0]^T, [\hat{\mathbf \Sigma}_g, \mathbf 0])$.
Then problem \eqref{equ:robmse2} becomes
\ben\label{equ:becomes}
\min_{\hat{\mathbf \Sigma}_g, \hat{\mathbf \Sigma}_{f}}\max_{\|\hat{\mathbf \Delta}\|_F\leq \varepsilon }&&\|\hat{\mathbf \Sigma}_g(\hat{\mathbf \Sigma}_{h}+\hat{\mathbf \Delta})\hat{\mathbf \Sigma}_{f}-\mathbf I\|_F^2+\sigma_{d}^2\|\hat{\mathbf \Sigma}_g\|_F^2,\nonumber\\
\text{s.t.}&&\text{tr}(\hat{\mathbf \Sigma}_{f}\hat{\mathbf \Sigma}_{f}^H)\leq P,
\enn
where $\hat{\mathbf \Delta}$ is the upper left $L\times L$ submatrix of $\mathbf \Delta$.
We will then show that there exist an optimal $\hat{\mathbf \Delta}$ in \eqref{equ:becomes} that is diagonal.

After some matrix manipulations and noticing the fact that the maximization of a convex function is achieved on the boundary \cite{JW_Pre-fixed}, the inner maximization of problem  \eqref{equ:becomes} can be transformed into the following problem
\ben\label{equ:transform}
\min_{\|{\bm \delta}\|= \varepsilon}{\bm \delta}^H(-\mathbf B^T\otimes \mathbf C) {\bm \delta}-2\mathfrak{R}\{\mathbf d^H{\bm \delta}\},
\enn
where ${\bm \delta}\triangleq \text{vec}(\hat{\mathbf \Delta})$,
$\mathbf C\triangleq \hat{\mathbf \Sigma}_g\hat{\mathbf \Sigma}_g^H$, $\mathbf B\triangleq \hat{\mathbf \Sigma}_{f}\hat{\mathbf \Sigma}_{f}^H$, and $\mathbf d\triangleq \text{vec}(\hat{\mathbf \Sigma}_g^H(\hat{\mathbf \Sigma}_g\hat{\mathbf \Sigma}_{h}\hat{\mathbf \Sigma}_{f}-\mathbf I)\hat{\mathbf \Sigma}_{f}^H)$.

By the result in \cite{JW_Pre-fixed}, ${\bm \delta}$ is a global minimizer of \eqref{equ:transform} if and only if there exists an $\omega$ such that
\ben
(-\mathbf B^T\otimes \mathbf C+\omega \mathbf I){\bm \delta}=\mathbf d, -\mathbf B^T\otimes \mathbf C+\omega\mathbf I\succeq \mathbf 0, \|{\bm \delta}\|=\varepsilon,\nonumber
\enn
which is equivalent to
\ben
\omega \hat{\mathbf \Delta}-\mathbf C \hat{\mathbf \Delta}\mathbf B&=&\hat{\mathbf \Sigma}_g^H(\hat{\mathbf \Sigma}_g\hat{\mathbf \Sigma}_{h}\hat{\mathbf \Sigma}_{f}-\mathbf I)\hat{\mathbf \Sigma}_{f}^H,\label{equ:satis1}\\
\text{tr}(\hat{\mathbf \Delta}\hat{\mathbf \Delta}^H)&=&\varepsilon^2, \label{equ:satis2}\\
\omega&\geq& \lambda_{\max}(\mathbf B^T\otimes \mathbf C). \label{equ:satis3}
\enn

Since both $\mathbf C$ and $\mathbf B$ are diagonal, \eqref{equ:satis1}-\eqref{equ:satis3} tells us that for any given $\hat{\mathbf \Sigma}_f$ and $\hat{\mathbf \Sigma}_g$,
there exists an optimal $\hat{\mathbf \Delta}$ that is diagonal.
Denote the optimal solution of \eqref{equ:becomes} as $(\hat{\mathbf \Sigma}_{f}^\sharp, \hat{\mathbf \Sigma}_g^\sharp, \hat{\mathbf \Delta}^\sharp_{D})$ with $\hat{\mathbf \Delta}^\sharp_{D}$ being diagonal. To facilitate the analysis, we further define $\mathbf F'^\sharp\triangleq[\hat{\mathbf \Sigma}_{f}^\sharp, \mathbf 0]^T$, $\mathbf G'^\sharp\triangleq [\hat{\mathbf \Sigma}_{g}^\sharp, \mathbf 0]$ and $\mathbf \Delta_D^\sharp\triangleq \text{diag}\{\hat{\mathbf \Delta}^\sharp_{D}, \mathbf 0\}$.  Then by the  above definitions and discussions, we have
\ben\label{equ:robb7}
\max_{\|\hat{\mathbf \Delta}\|_F\leq \varepsilon }\text{MSE}(\mathbf F'^\sharp, \mathbf G'^\sharp, \mathbf \Delta)
=\text{MSE}(\mathbf F'^\sharp, \mathbf G'^\sharp, \mathbf \Delta^\sharp_D).
\enn

Now we discuss another particular situation when $\mathbf \Delta=\mathbf \Delta^\sharp_D$. Then problem~\eqref{equ:robmse2} becomes
\ben\label{equ:robmse4}
\min_{\mathbf F', \mathbf G'}&&\|\mathbf G'(\mathbf \Sigma_{h}+\mathbf \Delta^\sharp_D)\mathbf F' -\mathbf I\|_F^2
+\sigma_{d}^2\|\mathbf G'\|_F^2,\nonumber\\
\text{s.t.}&&\text{tr}(\mathbf F'\mathbf F'^H)\leq P.
\enn
which has been discussed in \cite{MD_MSE, perfect_mmse,NK_MMSE_Selection}, and the optimal solution is given by
\ben
(\mathbf F', \mathbf G')=([\hat{\mathbf \Sigma}_{f}, \mathbf 0]^T, [\hat{\mathbf \Sigma}_g, \mathbf 0]).\label{equ:directly}
\enn
Substituting \eqref{equ:directly} into problem \eqref{equ:robmse4}, we have
\ben\label{equ:robmselo}
\min_{\hat{\mathbf \Sigma}_f, \hat{\mathbf \Sigma}_g}&&\|\hat{\mathbf \Sigma}_g(\hat{\mathbf \Sigma}_{h}+\hat{\mathbf \Delta}_D^\sharp)\hat{\mathbf \Sigma}_f -\mathbf I\|_F^2
+\sigma_{d}^2\|\hat{\mathbf \Sigma}_g\|_F^2,\nonumber\\
\text{s.t.}&&\text{tr}(\hat{\mathbf \Sigma}_f\hat{\mathbf \Sigma}_f^H)\leq P.
\enn
Remember that the optimal solution of \eqref{equ:becomes} is denoted by  $(\hat{\mathbf \Sigma}_{f}^\sharp, \hat{\mathbf \Sigma}_g^\sharp, \hat{\mathbf \Delta}^\sharp_{D})$. Then it is easy to know that the optimal solution of  \eqref{equ:robmselo} is given by  $(\hat{\mathbf \Sigma}_{f}^\sharp, \hat{\mathbf \Sigma}_g^\sharp)$, which means that when the channel uncertainty is $\mathbf \Delta^\sharp_D=\text{diag}\{\hat{\mathbf \Delta}^\sharp_{D}, \mathbf 0\}$, the optimal $(\mathbf F', \mathbf G')$ of problem \eqref{equ:robmse4} is given by $([\hat{\mathbf \Sigma}^\sharp_{f},\mathbf 0]^T, [\hat{\mathbf \Sigma}^\sharp_g, \mathbf 0])$, or we have
\ben\label{equ:com1}
\text{MSE}(\mathbf F', \mathbf G', \mathbf \Delta^\sharp_D)
\geq\text{MSE}(\mathbf F'^\sharp, \mathbf G'^\sharp, \mathbf \Delta^\sharp_D).
\enn

We will next show that from these two special cases given in \eqref{equ:robb7} and \eqref{equ:com1}, the joint optimal structure of $\mathbf F'$ and $\mathbf G'$ can be obtained. This technique has also been used in \cite{thywork, HS_Worst_MR, HS_MMSE}, and is detailed as follows
\begin{multline}
\max_{\|\mathbf \Delta\|_F\leq \varepsilon }\text{MSE}(\mathbf F', \mathbf G', \mathbf \Delta)
\overset{(a)}\geq\text{MSE}(\mathbf F', \mathbf G', \mathbf \Delta^\sharp_D)\nonumber\\
\overset{(b)}\geq\text{MSE}(\mathbf F'^\sharp, \mathbf G'^\sharp,  \mathbf \Delta^\sharp_D)
\overset{(c)}=\max_{\|\mathbf \Delta\|_F\leq \varepsilon }\text{MSE}(\mathbf F'^\sharp, \mathbf G'^\sharp, \mathbf \Delta)\yesnumber\label{equ:last}
\end{multline}
where $(a)$ is due to the fact that $\mathbf \Delta_D^\sharp$ is only a particular channel, $(b)$ is due to \eqref{equ:com1} and $(c)$ is due to \eqref{equ:robb7}. Inequality \eqref{equ:last} shows that the optimal $\mathbf F'$ and $\mathbf G'$ must be given by $(\mathbf F'^\sharp, \mathbf G'^\sharp)$. The proof is completed.
\end{proof}
\begin{remark}
Theorem~\ref{theorem:robmse} provides some interesting insights into the robust optimal transceiver design, showing that its optimal structure is given by $\mathbf V_f=\mathbf U_g$, $\mathbf U_f=\mathbf V_h$ and $\mathbf V_g=\mathbf U_h$, which diagonalizes the MIMO channel into eigen subchannels, and is consistent with the results under perfect and stochastic CSI  assumptions. Therefore, the answer to question $\mathcal Q1$ is yes! Note that problem~\eqref{equ:robmse} was also considered in \cite{JW_2011}, where only partial results of Theorem~\ref{theorem:robmse} was obtained. That is, by assuming $\mathbf V_g=\mathbf U_h$, then  $\mathbf V_f=\mathbf U_g$ and $\mathbf U_f=\mathbf V_h$ were proved to be optimal; on the other hand, given $\mathbf U_f=\mathbf V_h$, then $\mathbf V_f=\mathbf U_g$ and $\mathbf V_g=\mathbf U_h$ were proved to be optimal.
\end{remark}

\section{Robust Global Optimal Design Based On Scalar Optimization}\label{sec:scalar}
Based on Theorem~\ref{theorem:robmse}, we know that the optimal solution of
problem~\eqref{equ:robmse} is determined by problem~\eqref{equ:becomes}, where $\hat{\mathbf \Delta }$ is diagonal. Denote $\hat{\mathbf \Delta }\triangleq \text{diag}\{x_1, \cdots, x_L\}$, $\mathbf f\triangleq [f_1, \cdots, f_L]^T$ and $\mathbf g\triangleq [g_1, \cdots, g_L]^T$, then problem~\eqref{equ:becomes} is equivalent to
\ben
\min_{\mathbf f, \mathbf g}\max_{\sum_{i=1}^Lx_i^2\leq \varepsilon^2}&&\sum_{i=1}^L(f_ig_i(\gamma_i+x_i)-1)^2+\sigma_n^2\sum_{i=1}^Lg_i^2\nonumber\\
\text{s.t.}&&\sum_{i=1}^Lf_i^2\leq P.\label{equ:shall2}
\enn

Introducing a slack variable $t$, problem \eqref{equ:shall2} can be converted to
\begin{subequations}\label{equ:pen}
\ben
\min_{\mathbf f, \mathbf g, t}&&t+\sigma_n^2\sum_{i=1}^Lg_i^2\\
\text{s.t.}&&\sum_{i=1}^L(f_ig_i(\gamma_i+x_i)-1)^2\leq t, \sum_{i=1}^Lx_i^2\leq \varepsilon^2\label{equ:heroa}\\
&&\sum_{i=1}^Lf_i^2\leq P.
\enn
\end{subequations}

Generally speaking, it is difficult to derive a closed-form solution of \eqref{equ:pen}. Thus we will solve it in numerical results.
Let ${\bm\eta}\triangleq[f_1g_1\gamma_1-1, \cdots,f_Lg_L\gamma_L-1]^T$, $\mathbf \Gamma\triangleq \text{diag}\{f_1g_1 , \cdots,f_Lg_L\}$, and $\mathbf x\triangleq [x_1, \cdots, x_L]^T$.
Constraint \eqref{equ:heroa} can be rewritten as
\ben\label{equ:add1}
\|{\bm\eta}+\mathbf {\Gamma}\mathbf x\|_2^2\leq t, \|\mathbf x\|_2\leq \varepsilon.
\enn
Following the similar lines in \cite{JW_2011,JW_Pre-fixed},  one can transform \eqref{equ:add1} into
\ben\label{equ:schur}
\begin{bmatrix}t-\mu&&{\bm\eta}^H && \mathbf 0\\{\bm\eta}&&\mathbf I&&\varepsilon\mathbf \Gamma\\\mathbf 0&&\varepsilon\mathbf \Gamma^H&&\mu\mathbf I\end{bmatrix}\succeq \mathbf 0, \exists \mu\geq 0.
\enn

By applying Schur's Complement \cite{convex}, \eqref{equ:schur} is equivalent to the following constraint
\begin{multline}
\begin{bmatrix}t-\mu && {\bm \eta}^H \\{\bm \eta} && \mathbf I\end{bmatrix}-\frac{1}{\mu}\begin{bmatrix}\mathbf 0\\ \varepsilon\mathbf \Gamma\end{bmatrix}\begin{bmatrix}\mathbf 0&& \varepsilon\mathbf \Gamma^H\end{bmatrix}\nonumber\\
=\begin{bmatrix}t-\mu && {\bm \eta}^H \\{\bm \eta} && \mathbf I-\frac{1}{\mu}\varepsilon^2\mathbf \Gamma\mathbf \Gamma^H\end{bmatrix}\succeq \mathbf 0.\yesnumber\label{equ:schuragain}
\end{multline}
Using Schur's Complement again, \eqref{equ:schuragain} can be written as
\ben
t-\mu-{\bm \eta}^H(\mathbf I-\frac{1}{\mu}\varepsilon^2\mathbf \Gamma\mathbf \Gamma^H)^{-1}{\bm \eta}\geq \mathbf 0,\nonumber
\enn
or equivalently
\ben\label{equ:combine}
\mu+\sum_{i=1}^L\frac{(f_ig_i\gamma_i-1)^2}{1-\varepsilon^2f_i^2g_i^2/\mu}\leq t.
\enn
Combining \eqref{equ:pen} and \eqref{equ:combine}, we get the following problem
\ben\label{equ:importance}
\min_{\mathbf f, \mathbf g, \mu}&& \sum_{i=1}^L\frac{(\gamma_if_ig_i-1)^2}{1-\varepsilon^2f_i^2g_i^2/\mu}+\mu+\sigma_n^2\sum_{i=1}^Lg^2_i\nonumber\\
\text{s.t.}&&\sum_{i=1}^Lf_i^2\leq P,\,\,\mu\geq \varepsilon^2f_i^2g_i^2,
\enn
where the constraint $\mu\geq \varepsilon^2f_i^2g_i^2$ is implicitly included in the constraint \eqref{equ:schuragain}.


Problem \eqref{equ:importance} is still difficult to deal with. However,
we will next show that by some variable transformations, the global optimal solution of problem~\eqref{equ:importance} can be obtained. Define $\mathbf m\triangleq [m_1, \cdots, m_L]^T$ and $\mathbf n\triangleq [n_1, \cdots, n_L]^T$, where $m_i=f_ig_i$ and $n_i=g_i^2$, for $i=1, \cdots, L$. Then problem~\eqref{equ:importance} becomes
\ben\label{equ:importance2}
\min_{\mathbf m, \mathbf n, \mu}&&\phi_1(\mathbf m, \mathbf n, \mu)\triangleq \sum_{i=1}^L\frac{(\gamma_im_i-1)^2}{1-\varepsilon^2m_i^2/\mu}+\mu+\sigma_n^2\sum_{i=1}^Ln_i\nonumber\\
\text{s.t.}&&\sum_{i=1}^Lm_i^2/n_i\leq P,\,\,\mu\geq \varepsilon^2m_i^2,
\enn
Let $\mathbf s\triangleq [s_1, \cdots, s_L]^T$. We claim that \eqref{equ:importance2} is equivalent to the following problem
\ben\label{equ:jw2}
\min_{\mathbf m, \mathbf n,\mu, \mathbf s} &&\phi_2(\mathbf m, \mathbf n, \mu, \mathbf s)\triangleq \sum_{i=1}^L\frac{(\gamma_im_i-1)^2}{1-s_i}+\mu+\sigma_n^2\sum_{i=1}^Ln_i\nonumber\\
\text{s.t.}&&\sum_{i=1}^Lm_i^2/n_i\leq P, \,\,\varepsilon^2 m_i^2/\mu\leq s_i<1.
\enn

This can be explained as follows. First, suppose that $(\mathbf m^\sharp, \mathbf n^\sharp, \mu^\sharp, \mathbf s^\sharp)$ is the optimal solution of \eqref{equ:jw2}. In view of \eqref{equ:importance2}, it follows that
\begin{multline}
\min_{\mathbf m,\mathbf n,\mu}\phi_1(\mathbf m, \mathbf n, \mu)\leq\phi_1(\mathbf m^\sharp, \mathbf n^\sharp, \mu^\sharp)\\
\overset{(a)}\leq \phi_2(\mathbf m^\sharp, \mathbf n^\sharp, \mu^\sharp, \mathbf s^\sharp)=\min\phi_2(\mathbf m, \mathbf n, \mu, \mathbf s),\nonumber
\end{multline}
where in $(a)$ we have used the constraint $\varepsilon^2m_i^{\sharp 2}/\mu^{\sharp} \leq s_i^{\sharp}$.
Then we know that $\min\phi_1(\mathbf m, \mathbf n, \mu)\leq\min\phi_2(\mathbf m, \mathbf n, \mu, \mathbf s)$.
On the other hand, for any feasible $(\mathbf m, \mathbf n, \mu)$ of \eqref{equ:importance2}, we can always find some $\mathbf s$, which makes the equality holds in  \eqref{equ:jw2}. This means that, it is also feasible to \eqref{equ:jw2}. Then we must have  $\min\phi_1(\mathbf m, \mathbf n, \mu)\geq\min\phi_2(\mathbf m, \mathbf n, \mu, \mathbf s)$.
Thereby we must have $\min\phi_1(\mathbf m, \mathbf n, \mu)=\min\phi_2(\mathbf m, \mathbf n, \mu, \mathbf s)$}.

Introducing the slack variable $\mathbf z\triangleq [z_1, \cdots, z_L]^T$, problem \eqref{equ:jw2} can be written as
\ben\label{equ:new3}
\min_{\mathbf m, \mathbf n, \mu, \mathbf s, \mathbf z} &&\sum_{i=1}^Lz_i+\mu+\sigma_n^2\sum_{i=1}^Ln_i\nonumber\\
\text{s.t.}&&\begin{bmatrix}z_i && \gamma_im_i-1 \\\gamma_im_i-1&& 1-s_i\end{bmatrix}\succeq \mathbf 0, \,\, i=1, \cdots, L, \nonumber\\
&&\begin{bmatrix}\mu && \varepsilon m_i\\\varepsilon m_i&&s_i\end{bmatrix}\succeq \mathbf 0,\,\,  s_i<1, \,\,i=1, \cdots, L,\nonumber\\
&&\sum_{i=1}^L\frac{m_i^2}{n_i}\leq P.
\enn

Problem \eqref{equ:new3} is a semidefinite programming (SDP) problem, which can be efficiently solved by the MATLAB package tools such as CVX \cite{CVX}. Then the optimal $f_i$ and $g_i$ are obtained by $f_i=m_i/\sqrt{n_i}$ and $g_i=\sqrt{n_i}$.

\begin{remark}
By fixing $\mathbf f$ or $\mathbf g$, problem~\eqref{equ:importance}  becomes problem ($6$) or ($7$) in \cite{JW_2011}, where they were proved to be convex and can be optimally solved, respectively. This process is repeated until convergence.
However, the solution of this alternating
optimization based method depends on the initial point $\mathbf f^{(0)}$, and may not be optimal if problem~\eqref{equ:importance} has local minimal. Thus it is natural to ask question $\mathcal Q2$, does this method converge to the global optimal solution?  Unfortunately, this is not guaranteed. As will be seen in section~\ref{sec:simu}, different initial point $\mathbf f^{(0)}$ will lead to different results,  and may also incur some performance loss. Therefore, our answer to question $\mathcal Q2$ is \emph{not always, and it depends on the initial point}.
\end{remark}

\section{Discussions and Simulations}\label{sec:simu}
In this section, we first provide the complexity comparison between our joint optimal design and  the robust design in \cite{JW_2011}, and then  give  numerical results to compare the two robust designs.

The channel fading is modeled as Rayleigh fading, and each channel entry satisfies the complex normal distribution $\mathcal C\mathcal N(0,1)$. The noise is assumed to be zero-mean unit variance complex Gaussian random variables. In our simulations, we set $M=N=L$ and vary $\varepsilon$ through the normalized parameter $\rho\in [0, 1)$, i.e., $\varepsilon^2=\rho\|\tilde {\mathbf H}\|_F^2$. Then the larger the $\rho$ is, the poorer the CSI quality will be. All results are averaged over $1000$ channel realizations.
\subsection{Complexity Comparison}
Since we have derived the optimal structure of tranceiver design in Theorem~\ref{theorem:robmse}, the major computing step in our work remains in solving problem~\eqref{equ:new3}.
The complexity for solving the SDP problem~\eqref{equ:new3} is $\mathcal O(L^2)$ per iteration  and the number of iterations typically lies between $5$ and $50$, for an SDP problem \cite{boyd_sdp}. On the other hand, the complexity analysis of the method in \cite{JW_2011} is a little complicated. In their work, when fixing $\mathbf g$, the optimal $\mathbf f$ was obtained by the three-level primal-primal decomposition method, where a close-form solution was given in the lowest level, while the bisection method and the gradient method were applied at the middle and third level, respectively. Similarly, the problem for determining optimal $\mathbf g$ under fixed $\mathbf f$ was also solved in two levels, where a close-form solution was derived at the first level, while the bisection method was used at the second level. It can be seen that it is hard to determine the complexity of each iteration as well as the exact (or even approximate) iteration number in this method. Upon this observation, we resort to the CPU time comparison required by the two methods.

Fig.~\ref{fig:cpu_time_vs_N} shows the average CPU time comparison between our robust optimal method and the robust method in \cite{JW_2011}. We set $P=20$dBW and choose different initial points for the method in \cite{JW_2011}: a) Scheme I: set the initial point $\mathbf f^{(0)}$ with equal elements. b) Scheme II: set $\mathbf f^{(0)}$ as the non-robust solution that takes the nominal channel $\tilde {\mathbf H}$ as the actual channel \cite{perfect_mmse}. c) Scheme III: set $\mathbf f^{(0)}$ as a random variable that satisfies $\|\mathbf f^{(0)}\|_2^2=P$. It can be observed from Fig.~\ref{fig:cpu_time_vs_N} that scheme III is the most time-consuming scheme among all the schemes. Although the time cost by Scheme I and II is similar to that in our our global optimal solution, it increases more rapidly as the number of data streams L goes large.
  From Fig.~\ref{fig:cpu_time_vs_N}, we know that the required CPU time for the method in \cite{JW_2011} tends to be more random in nature, and heavily depends on the initial point. On the other hand, our method is not only efficient but also  has a more stable runtime performance.

\subsection{Numerical Results}
\begin{figure}[!t]
    \centering
    \includegraphics[width=3.5in]{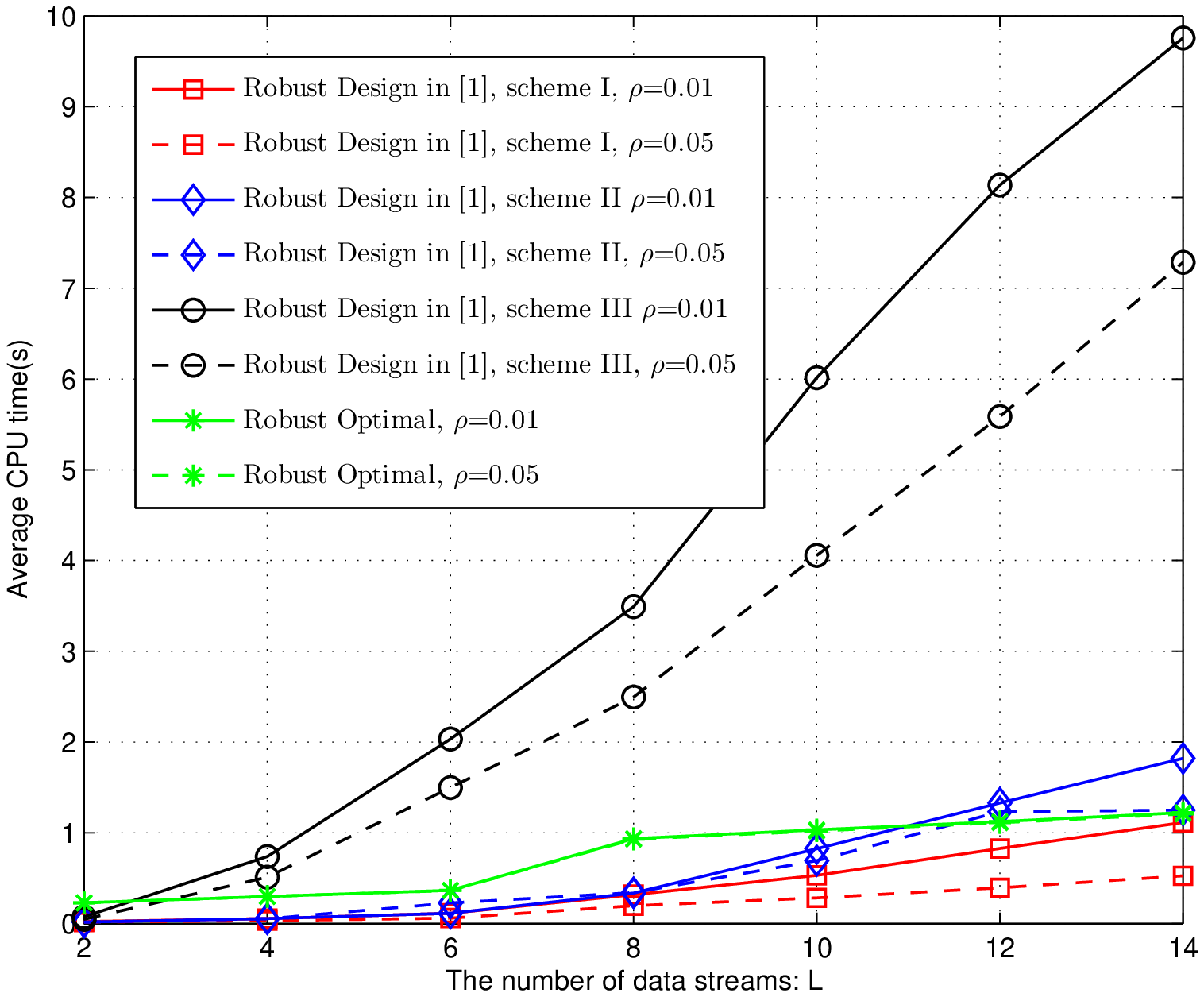}
    \caption{Average CPU time comparison versus different $L$ for $M=N=L$ and $P=20$dBW.}\label{fig:cpu_time_vs_N}
%
    \includegraphics[width=3.5in]{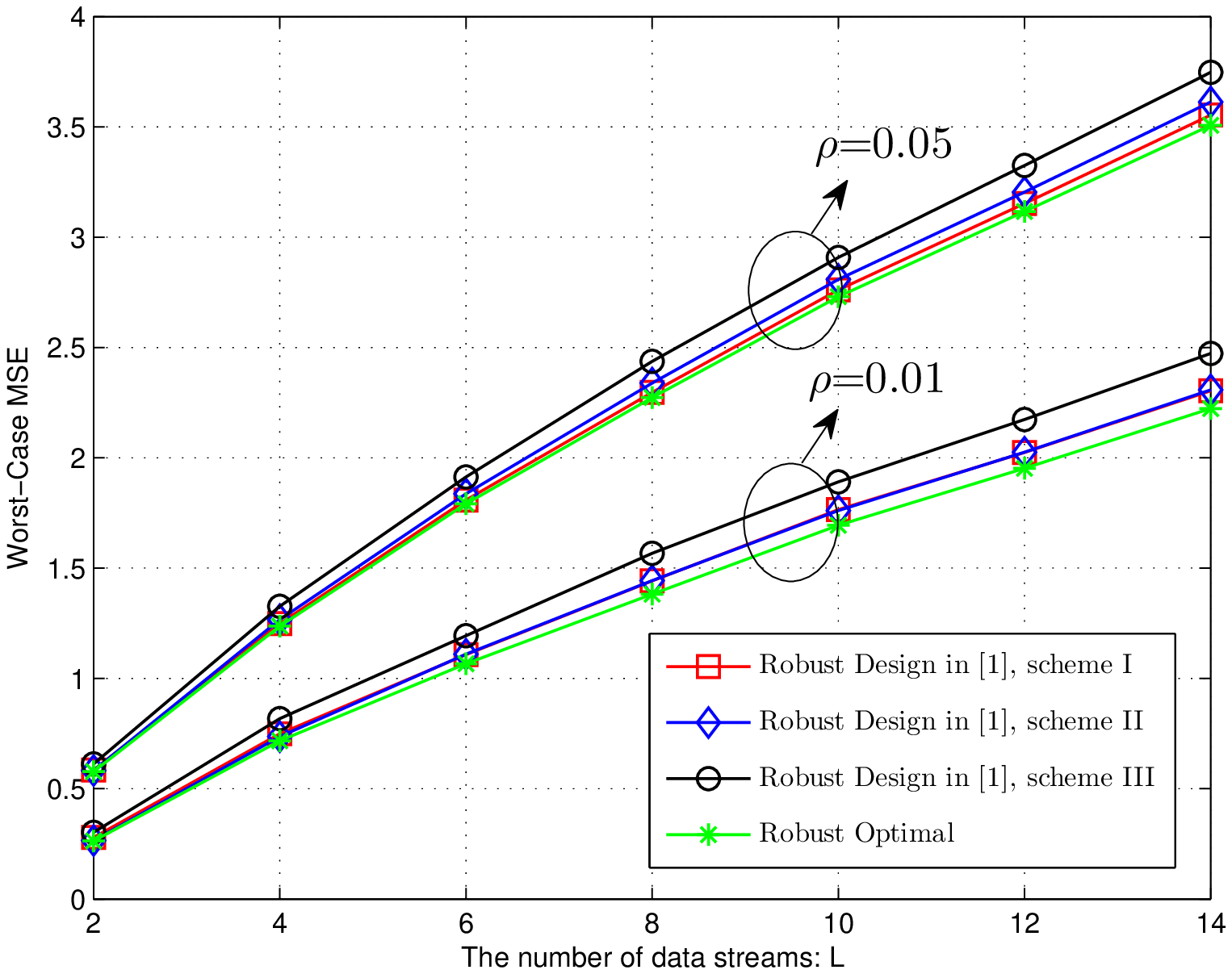}
    \caption{Worst-case MSE versus different $L$ for $M=N=L$ and $P=20$dBW.}\label{fig:mimorho}
\end{figure}
\begin{figure}[!t]
    \centering
    \includegraphics[width=3.5in]{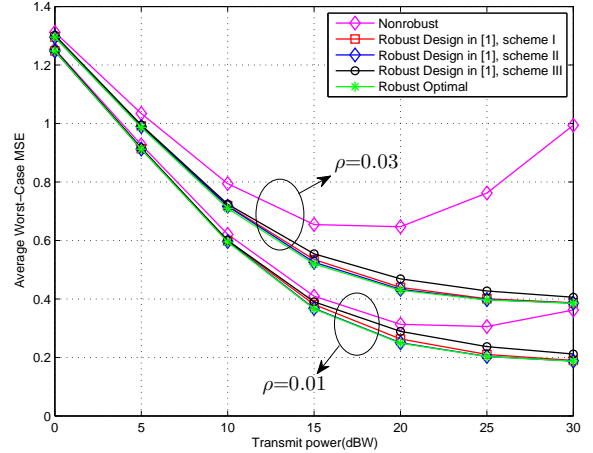}
    \caption{Worst-case MSE versus different transmit power for $M=N=L=2$.}\label{fig:mimosnr}
\end{figure}
We now study the system MSE performance in different scenarios.
In Fig.~\ref{fig:mimorho}, we set the same network configuration as that in Fig.~\ref{fig:cpu_time_vs_N}, and investigate the average worst-case MSE performance versus $L$ of our method and three different schemes in \cite{JW_2011}. As shown in the plot, scheme I and scheme II suffer some marginal performance loss, while scheme III  incurs some apparent loss, which grows even larger when $L$ increases. Hence, the method in \cite{JW_2011} does not always lead to the optimal solution.

As an other example,
Fig.~\ref{fig:mimosnr} depicts the average worst-case MSE performance versus different transmit power with $\rho=0.01$ and $\rho=0.03$. We consider the case when $M=N=L=2$. Fig.~\ref{fig:mimosnr} verifies the superior performance of robust schemes over the non-robust scheme in \cite{perfect_mmse}. Moreover, it can be observed  that scheme II  has an almost optimal performance, while scheme I approximately approaches to the optimal solution. Therefore, when the advanced software
package (such as CVX) is not available, scheme I (or scheme II) can  be viewed as a simple implementation of the global optimal method.

\section{Conclusions and Future Work}
In this paper, we investigate the global optimal transceiver design in the MIMO link under deterministic  CSI uncertainty model. We first prove that
the optimal design of the transmitter-equalizer has a favorable channel-diagonalizing structure. Then we simplify the original problem into a scalar optimization problem, and obtain the global optimal solution via an SDP problem. Simulation results show that our method outperforms the existing schemes.

We only considered the point-to-point MIMO system in this work. However, as pointed out in \cite{EB_color}, the transmission can also  be affected by the multiple access interference (MAI), if the transmit-receive nodes are active over a communication network which employs non-orthogonal multiplexing. In this case, the  received signal at the destination is contaminated by  the spatially colored Gaussian noise and the robust optimal transceiver design must be reconsidered. Hence, it would be interesting to
address this issue in our future  research.
%

\end{document}